\documentclass[amsmath, aps, prx, twocolumn, notitlepage]{revtex4-1}

\usepackage{graphicx}
\usepackage{dcolumn}
\usepackage{bm}
\usepackage[usenames]{color}
\usepackage{slashed}
\usepackage[utf8]{inputenc}
\usepackage{amsfonts}
\usepackage{amsmath}
\usepackage{amssymb}
\usepackage{amsthm}
\usepackage{hyperref}
\usepackage{latexsym}
\usepackage{color}
\usepackage{enumitem}
\makeatletter
\def\Let@{\def\\{\notag\math@cr}}
\makeatother

\newtheorem{theorem}{Theorem}
\newcommand{\sddots}{\raisebox{0pt}{$\scalebox{.6}{$\ddots$}$}}

\begin{document}

\title{Real spectra in one-dimensional single-band non-Hermitian Hamiltonians}
\author{Haoyan Chen}
\email{chenhaoyan@stu.pku.edu.cn}
\affiliation{International Center for Quantum Materials, School of Physics, Peking University, Beijing 100871, China}
\author{Yi Zhang}
\email{frankzhangyi@gmail.com}
\affiliation{International Center for Quantum Materials, School of Physics, Peking University, Beijing 100871, China}
\date{\today}

\begin{abstract}
In general, the energy spectrum of a non-Hermitian system turns out to be complex, which is not so satisfactory since the time evolution of eigenstates with complex eigenvalues is either exponentially growing or decaying. Here, we provide a sufficient and necessary condition of the real spectrum under open boundary conditions for one-dimensional non-Hermitian tight-binding Hamiltonians. The necessity is directly related to the fact that the generalized Brillouin zone in one dimension is a closed loop with the origin in its interior. We also establish the sufficiency by analytically determining when the preimage of the characteristic polynomial contains a loop and showing that this loop is just the generalized Brillouin zone itself, in some simple models first and then general one-band models. We demonstrate our conclusions on various non-Hermitian model examples with longer-range hopping. Our results indicate that real spectra are more common than one may have expected in non-Hermitian systems and are helpful for designing non-Hermitian models with real spectra.
\end{abstract}

\maketitle

\section{Introduction}

Non-Hermitian Hamiltonians can offer effective descriptions of open systems interacting with the external environment, such as the localization of metal with impurities \cite{Hatano1996, Hatano1997, Hatano1998} and the radioactive decay in nuclear reactions \cite{Feshbach1958, Feshbach1962}. Interestingly, non-Hermiticity gives rise to fascinating concepts and phenomena dramatically different from Hermitian systems \cite{Emil2021, Yong2017, Zhesen2019, Ashida2020, Kohei2019, Kohei2019symmetry, Zhou2019, Carlstrom2019, Kawabata2022, Xue2022, Rui2022, Matsumoto2020, Liu2019, Longhi2022, Kawabata2019, Edvardsson2019, Leykam2017, Shen2018}. In particular, the energy spectrum of a non-Hermitian Hamiltonian will generally be complex under the periodic boundary condition (PBC). When the non-Hermitian winding number is nonzero \cite{Zhang2020, Nobuyuki2020,Borgnia2020}, the non-Hermitian skin effect occurs, i.e., extensive eigenstates become localized at the system's boundary under the open boundary condition (OBC), and the energy spectrum collapses into analytic curves inside the region encircled by the periodic boundary spectrum \cite{Yao2018, Kohei2020, Zhang2022, Song2019, Longhi2021, Claes2021, Gong2018}. This sensitivity to boundary conditions leads to the failure of the conventional bulk-boundary correspondence in non-Hermitian systems \cite{Hasan2010, Qi2011, Lee2016, Guo2021, Kunst2018, Xiong2018, Herviou2018, Song2019b, Xiao2020}. Generally, the energy spectrum of a 1D non-Hermitian system under OBC is obtainable via the generalized Brillouin zone (GBZ) method \cite{Palle1960, Bottcher2005, Zhang2020, Zhesen2020, Kazuki2019, Yao2018, Okuma2022}. However, this method may become invalid in higher dimensions in case of open boundary conditions along more than one direction \cite{Kohei2020}.

Despite its attractive properties, non-Hermiticity breaks some axioms of quantum mechanics, such as the unitarity of time evolution required for probability conservation. Thus, non-Hermitian Hamiltonians are mostly phenomenological instead of fundamental descriptions of physical systems. On the other hand, a purely real spectrum, e.g., given certain symmetries of the Hamiltonian, may recover time evolution's unitarity. Besides, complex spectra in non-Hermitian systems may lead to ambiguity in the definition of ground states \cite{guo2022variational}: the state whose real part of the energy is minimal, which is conventional and consistent with Hermitian quantum systems, and the state whose imaginary part is maximal, which grows fastest under time evolution and dominates the long-time behavior of the system. Such ambiguity disappears for non-Hermitian systems exhibiting real spectra, which offers a solid foundation for many-body physics in non-Hermitian systems. For example, the modified harmonic oscillator Hamiltonian $H=p^2+x^2(ix)^\epsilon$ is $\mathcal{PT}$-symmetric when $\epsilon$ is real. Then, the eigenvalues are all real if such $\mathcal{PT}$ symmetry is not spontaneously broken, as observed for $\epsilon\geq 0$ \cite{Bender1998, Bender2007}. Non-Hermitian 2D topological insulators can also have real spectra if they preserve a variant of time-reversal symmetry and the pseudo-Hermiticity condition \cite{Kawabata2020}. However, it is currently unclear whether non-Hermitian systems can have real spectra without symmetry constraints and, if so, the corresponding conditions. 

Here, we show that real spectra can be achieved in 1D single-band non-Hermitian systems under OBC without symmetry protection: the preimage of the characteristic polynomial on the real axis must contain a closed loop that encircles the origin (we will call this the loop property for simplicity in the following). First, we show that the generalized Brillouin zone for a 1D single-band non-Hermitian tight-binding Hamiltonian is always a close loop in the complex plane with $0$ in its interior \cite{Zhang2020}. We demonstrate this by showing that any continuous curve that links the origin and infinity must intercept the generalized Brillouin zone. This implies that single-band GBZ in the complex plane has the loop property. By definition of GBZ, we conclude that the real OBC spectrum requires the preimage of the characteristic polynomial on the real axis to have the loop property. To show that the converse statement is also true, we first extend the famous Hatano-Nelson model to the complex case. By directly calculating the eigenvalues \cite{smith1985}, we show that the complex Hatano-Nelson model has a real OBC spectrum if and only if the multiple of two hopping coefficients is real and positive. Then we illustrate that this condition and the loop property of the preimage are identical. Next, we add a next-nearest hopping to the original Hatano-Nelson model. In this case, we give the condition where the loop property holds. This is accomplished by establishing the equivalence to three intersection points of the preimage on the real axis. Furthermore, we show that under this condition, the GBZ is just the loop in the preimage of the characteristic polynomial by tracing the flow of the roots as energy varies. By the definition of GBZ, the equivalence between the loop property and real OBC spectrum is also established. Finally, we apply such method to show the equivalence between the loop property of the preimage and the real spectrum under OBC for general 1D single-band non-Hermitian Hamiltonians.

\section{Generalized Brillouin zone in 1D non-Hermitian systems}

\subsection{Generalized Brillouin zone}

We begin by reviewing the theory of non-Bloch theory and generalized Brillouin zone for 1D non-Hermitian systems. The Hamiltonian of a single-band non-Hermitian tight-binding model in one dimension can be written as:
\begin{equation}
    H=\sum\limits_{m,n}t_{m-n}c_m^\dagger c_n=\sum\limits_k\sum\limits_{l}t_{l}e^{-ikl}c_k^\dagger c_k\equiv \sum\limits_k H(k)c_k^\dagger c_k, \label{Hamiltonian}
\end{equation}
where the Hamiltonian in the momentum space is:
\begin{equation}
    H(k)=\sum\limits_{n=-p}^{q}t_ne^{ikn},
\end{equation}
and $p$ ($q$) denotes the maximum hopping range to the left (right). For example, the Hatano-Nelson model with only nearest-neighbor hopping has $p=q=1$. To obtain the energy spectrum under OBC in the thermodynamic limit (namely, lattice size $N\rightarrow\infty$), we analytically extend the momentum $k$ from the real axis to the complex plane and define $\beta:=e^{ik}$. In this case, the trajectory of $\beta$ will vary from the conventional Brillouin zone, i.e., the unit circle \cite{Yao2018, Kohei2020}. Therefore the Hamiltonian $H(k)$ will become a characteristic polynomial $a(\beta)$:
\begin{equation}
    a(\beta):=H(e^{ik}\rightarrow \beta)=\sum\limits_{n=-p}^qt_n\beta^n, \label{characteristic polynomial}
\end{equation}
where we assume $p$, $q$ are positive integers and $a_{-p}$ and $a_q$ are both nonzero. Solve the characteristic equation $\det(H(\beta)-E)=0$ for arbitrary $E\in\mathbb{C}$, we then obtain $p+q$ roots $\beta_i(E), i=1,2,\cdots,p+q$. Arrange these roots according to their norms \cite{Kazuki2019}:
\begin{equation}
    |\beta_1(E)|\leq|\beta_2(E)|\leq|\beta_3(E)|\leq\cdots\leq|\beta_{p+q}(E)|, \label{sort}
\end{equation}
and $E$ belongs to the open boundary energy spectrum $\sigma(H_{\text{OBC}})$ if and only if:
\begin{equation}
    |\beta_p(E)|=|\beta_{p+1}(E)| \label{GBZ}.
\end{equation}

The generalized Brillouin zone is constituted by $\beta_p(E)$ and $\beta_{p+1}(E)$ that satisfy Eq. (\ref{GBZ}) \cite{Kazuki2019,Zhang2020,Zhesen2020}. For Hermitian systems, the generalized Brillouin zone obtained this way is just the conventional Brillouin zone $|\beta|=1$. 

\subsection{The Hatano-Nelson model}

One of the simplest examples of non-Hermitian Hamiltonian is the Hatano-Nelson model \cite{Hatano1996}: 
\begin{equation}
    H_{1}=\sum\limits_i(t_lc_{i+1}^\dagger c_i+t_rc_{i-1}^\dagger c_i). \label{Hatano_Nelson}
\end{equation}
We assume $t_l, t_r>0$ for simplicity. The corresponding characteristic equation is given by:
\begin{equation}
    a_{1}(\beta)=t_l\beta^{-1}+t_r\beta=E,
    \label{HN_characteristic_polynomial}
\end{equation}
whose two roots $\beta_{\pm}$ satisfy:
\begin{equation}
    \beta_++\beta_-=\frac{E}{t_r}, \quad \beta_+\beta_-=\frac{t_l}{t_r}.
\end{equation}
From $|\beta_+|=|\beta_-|$, we obtain:
\begin{equation}
    |\beta_+|=|\beta_-|=\sqrt{\Big|\frac{t_l}{t_r}\Big|}:=r. \label{GBZcondition}
\end{equation}
Let $\beta_+=re^{i\theta}, \beta_-=re^{i\theta'}$, then:
\begin{equation}
    \theta+\theta'=0, \quad E=2\sqrt{t_lt_r}\cos\theta, \quad \theta\in[0,\pi].
\end{equation}
In contrast to the complex PBC spectrum $E_{\text{PBC}}=t_le^{-ik}+t_re^{ik}$ that forms an ellipse in the complex plane, the OBC spectrum flattens into a line on the real axis.  

\subsection{Spectra under different boundary conditions}

The significant difference between PBC and OBC spectra indicates that non-Hermitian systems may be highly sensitive to boundary conditions. In fact, if we impose a semi-infinite boundary condition (SIBC) on a 1D non-Hermitian lattice, then $E$ belongs to the SIBC spectrum $\sigma_{\text{SIBC}}(H)$ if and only if $E$ belongs to the PBC spectrum $\sigma_{\text{PBC}}(H)$ or its non-Hermitian winding number $w(E)\neq 0$, where \footnote{This result is an application of the famous Atiyah-Singer theorem, which relates the analytical index of Fredholm operator to the topological index, i.e., winding number, in a special case.}:
\begin{equation}
    w(E)=\int_0^{2\pi}\frac{dk}{2\pi i}\frac{d}{dk}\log\det(H(k)-E).
\end{equation}
Namely, the SIBC spectrum is just the PBC spectrum together with the region it encloses that has non-zero winding number \cite{Nobuyuki2020}. It is also obvious that the OBC spectrum is contained in the SIBC spectrum:
\begin{equation}
    \sigma_{\text{OBC}}(H)\subset\sigma_{\text{SIBC}}(H),
\end{equation}
because the OBC is given by the SIBC together with an additional boundary condition at the other end \cite{Nobuyuki2020}. A less obvious conclusion is that the winding number of the OBC spectrum is always zero \cite{Nobuyuki2020, Zhang2020}. That means the OBC spectrum cannot contain a circular or elliptical structure - it must collapse into some curves inside the region enclosed by the PBC spectrum. The relations of the energy spectra under different boundary conditions are shown in Fig. \ref {HNmodel}(a). 

The unusual behavior of the non-Hermitian spectrum under different boundary conditions demonstrates that adding a small boundary perturbation term to a non-Hermitian system under OBC can dramatically change its spectrum as long as the lattice size $N$ is large enough (see Fig. \ref{HNmodel}(b) and Ref. \cite{Guo2021}). This can be seen from the $\epsilon$-pseudospectrum \cite{Okuma2022, Okuma2021, Okuma2020, Trefethen2005}, which is defined as the set of spectra of Hamiltonian plus a perturbed matrix whose 2-norm is bounded by $\epsilon$:
\begin{equation}
    \sigma_\epsilon(H)=\bigcup_{||\eta||\leq\epsilon}\sigma(H+\eta).
\end{equation}
Indeed, for open boundary non-Hermitian Hamiltonian with lattice size $N$, the following formula holds \cite{Okuma2022, Trefethen2005}:
\begin{equation}
    \lim\limits_{\epsilon\rightarrow 0}\lim\limits_{N\rightarrow\infty}\sigma_\epsilon(H_{\text{OBC}}^{(L)})=\sigma_{\text{SIBC}}(H),
\end{equation}
which explains that the OBC spectrum is unstable against any small boundary additional terms in the thermodynamic limit for the Hatano-Nelson model \cite{Guo2021}. 

The significant difference between PBC and OBC spectra also indicates that eigenstates under various boundary conditions may have different behaviors - a phenomenon called the non-Hermitian skin effect \cite{Yao2018, Zhang2020, Nobuyuki2020,Borgnia2020}. Namely, in the original Hermitian systems, the eigenstates (except for the edge states) are Bloch states composed of a plane wave and a periodic function \cite{ashcroft2011}. However, in non-Hermitian systems under open boundary conditions, the eigenstates may decay with respect to the distance from one end of the system, as shown in Fig. \ref{HNmodel}(c). In this case, the crystal momentum $k$ takes complex values, and GBZ $\beta\equiv e^{ik}$ differs from the conventional Brillouin zone - a unit circle. Nevertheless, the GBZ is still a closed loop, as demonstrated later.

Another aspect of non-Hermitian Hamiltonians' sensitivity toward the boundary condition lies in their imprecise linear-algebra numerics when the lattice size $N$ is large \cite{Zhesen2020, Colbrook2019}. Thus, instead of exact diagonalization, we use an algorithm based on the generalized Brillouin zone theory applicable for universal single-band 1D tight-binding non-Hermitian Hamiltonians, see Appendix \ref{compute_OBC}.

\begin{figure}[!ht]
    \centering
    \includegraphics[width=\linewidth]{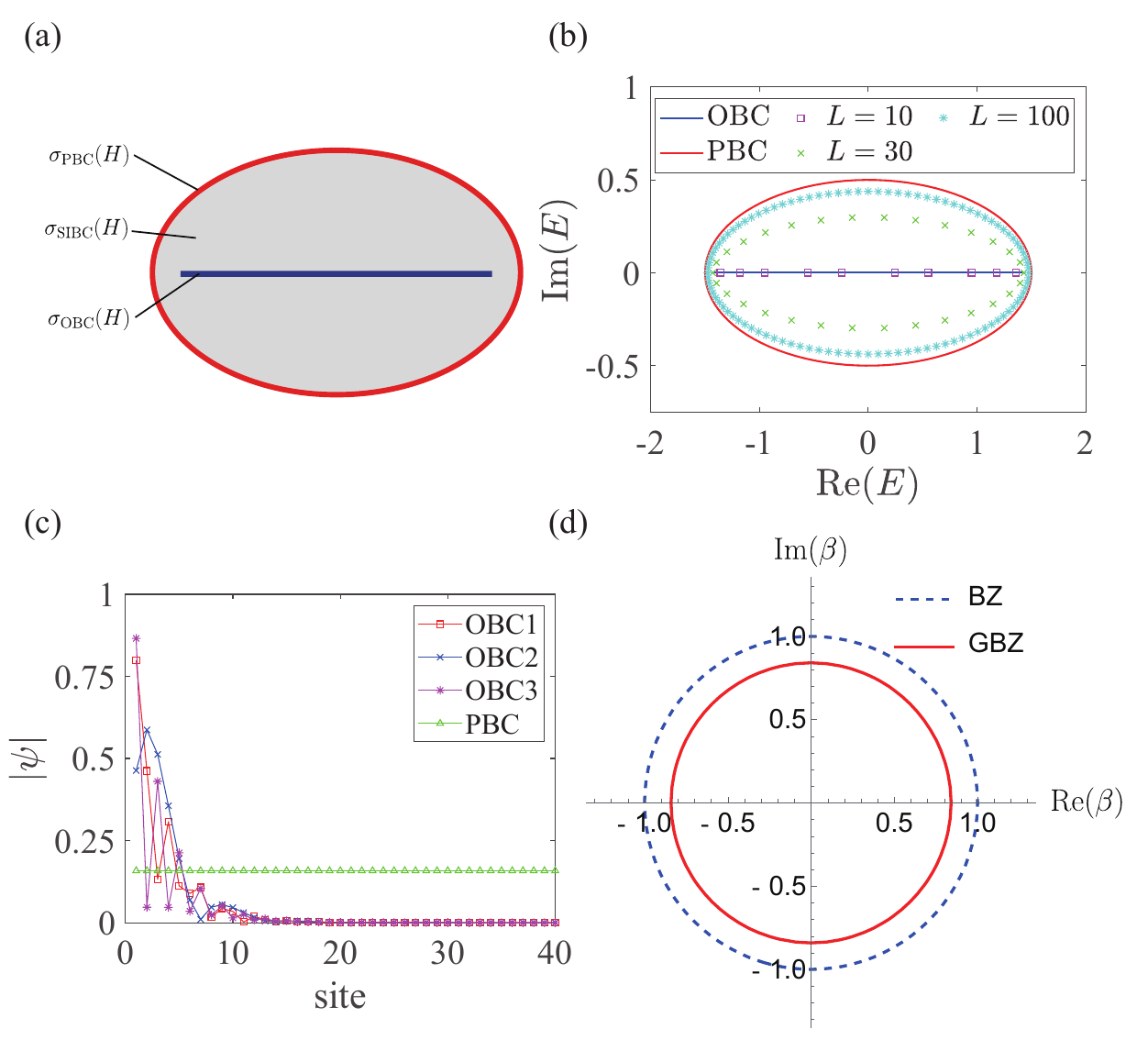}
    \caption{(a) Energy spectrum for the Hatano-Nelson model under PBC (red), SIBC (grey), and OBC (blue). (b) Eigenvalues for the Hatano-Nelson model  ($t_l=0.5$, $t_r=1$) under PBC, OBC, and a particular boundary condition with $0.01c_L^\dagger c_1+0.01c_1^\dagger c_L$ and various lattice sizes $N=10, 30, 100$. (c) The amplitude distributions for three randomly chosen eigenstates of the Hatano-Nelson model ($t_l=0.5$, $t_r=1$, $L=40$) under OBC illustrate their decay away from the boundary, while the eigenstates under PBC are still Bloch states. (d) The conventional Brillouin zone of the Hermitian systems and the generalized Brillouin zone of the non-Hermitian Hatano-Nelson model with $t_l=0.5$ and $t_r=1$.}
    \label{HNmodel}
\end{figure}

\subsection{Properties of the GBZ}

In general, one can show that for a general single-band Hamiltonian in Eq. (\ref{Hamiltonian}), the GBZ is always a closed curve that encircles the origin of the complex plane \cite{Zhang2020, Shapiro2019}. In fact, for an arbitrary continuous curve $\gamma(t): [0,1]\rightarrow\mathbb{C}$ that connects $z=0$ and $z=\infty$:
\begin{equation}
    \gamma(0)=0, \quad \lim_{t\rightarrow 1^-}|\gamma(t)|=\infty.
\end{equation}

Note that $\gamma(t)$ is a solution of the characteristic equation when $E=a(\gamma(t))$:
\begin{equation}
    \det(H(e^{ik}\rightarrow\beta)-E)=a(\beta)-E=a(\beta)-a(\gamma(t))=0. \label{E=a(gamma(t))}
\end{equation}
Therefore:
\begin{equation}
    \gamma(t)\in\{\beta_1(a(\gamma(t))),\beta_2(a(\gamma(t))),\cdots,\beta_{p+q}(a(\gamma(t)))\},
\end{equation}
where $\beta_i(a(\gamma(t))), i=1,2,\cdots,p+q$ are roots of Eq. (\ref{E=a(gamma(t))}) which are arranged as in Eq. (\ref{sort}). The first term $a_{-p}\beta^{-p}$ and the last term $a_q\beta^q$ of the characteristic polynomial (\ref{characteristic polynomial}) indicate that for $|E|\rightarrow\infty$ we have:
\begin{eqnarray}
    |\beta_1(E)|,|\beta_2(E)|,\cdots,|\beta_p(E)|&\rightarrow& 0, \nonumber\\
    |\beta_{p+1}(E)|,|\beta_{p+2}(E)|,\cdots,|\beta_{p+q}(E)|&\rightarrow& \infty, \label{root_limit}
\end{eqnarray}
When $t\rightarrow 0^+$, $|\gamma(t)|\rightarrow 0$, $|a(\gamma(t))|\rightarrow\infty$, then $\gamma(t)$ must appear among the first $p$ roots:
\begin{equation}
    \gamma(t)\in\{\beta_1(a(\gamma(t))),\beta_2(a(\gamma(t))),\cdots,\beta_p(a(\gamma(t)))\};
\end{equation} 
when $t\rightarrow 1^-$, on the other hand, $|\gamma(t)|\rightarrow\infty$, $|a(\gamma(t))|\rightarrow\infty$, $\gamma(t)$ appears among the last $q$ roots:
\begin{equation}
    \gamma(t)\in\{\beta_{p+1}(a(\gamma(t))),\beta_{p+2}(a(\gamma(t))),\cdots,\beta_{p+q}(a(\gamma(t)))\}.
\end{equation}

Since $\gamma(t)$ is a continuous curve, there must exist a critical point $t=t_c$ between these two cases as $t$ goes from $0$ to $1$. At this critical point $t_c\in (0, 1)$, we have:
\begin{eqnarray}
    \gamma(t_c)&\in&\{\beta_p(a(\gamma(t_c))),\beta_{p+1}(a(\gamma(t_c)))\}, \nonumber\\ 
    |\gamma(t_c)|&=&|\beta_p(a(\gamma(t_c)))|=|\beta_{p+1}(a(\gamma(t_c)))|, \label{gamma2}
\end{eqnarray}
which implies that $\gamma(t_c)$ belongs to the GBZ. Since $\gamma(t)$ is an arbitrary curve connecting $0$ and $\infty$, and since the OBC spectrum is always connected \cite{Ullman1967}, the GBZ should also be connected, and we conclude that the GBZ for 1D single-band non-Hermitian Hamiltonian in Eq. (\ref{Hamiltonian}) is a closed loop with $z=0$ in its interior. For example, the GBZ of the Hatano-Nelson model is in Eq. (\ref{GBZcondition}) and essentially a circle, as shown in Fig. \ref{HNmodel}(d).

These results imply that if the energy spectrum of a single-band non-Hermitian Hamiltonian is real under OBC, the preimage of $\mathbb{R}$ under its characteristic polynomial $a$, denoted as $a^{-1}(\mathbb{R})$ and contains the GBZ, must have the loop property (contain a loop that encircle $0$). Thus, we have obtained a necessary condition for real spectra in 1D non-Hermitian systems under OBCs. Here, the preimage is a mathematical concept - for a given function $f:X \rightarrow Y$ and a subset $B\subseteq Y$, the preimage of $B$ under the function $f^{-1}(B)$ is the set:
\begin{equation}
    f^{-1}(B)=\{x\in X:f(x)\in B\}.
\end{equation}
Namely, the preimage describes the ``inverse" of a function. In our case, the function refers to the characteristic polynomials of the non-Hermitian models, and the subset $B$ is the real axis.

Then, one may wonder whether such a condition is sufficient. In the following, we will consider the complex-hopping Hatano-Nelson model and the Hatano-Nelson model with a next-nearest-neighbor hopping term and demonstrate the condition's sufficiency in these two models analytically. Further, we extend the method for the Hatano-Nelson model with a next-nearest-neighbor hopping to show the sufficiency in general 1D single-band non-Hermitian models. We also refer the audience to Appendix \ref{example} for additional numerical examples.

\section{Applications and Analysis for model examples} \label{NNNhopping}

\subsection{The Hatano-Nelson model with complex hopping}

First, we revisit the Hatano-Nelson model in Eq. (\ref{Hatano_Nelson}), whose characteristic polynomial $a_1(\beta)$ is given in Eq. (\ref{HN_characteristic_polynomial}). Under a similarity transformation \cite{Hatano1996, Hatano1997, Kohei2020, Okuma2022, Rivero2021}:
\begin{equation}
    U_r^{-1} c_i^\dagger U_r=r^i c_i^\dagger, \quad U_r^{-1} c_i U_r=r^{-i} c_i, \quad r>0, \label{similar}
\end{equation}
which preserves the energy spectrum of the Hamiltonian in Eq. (\ref{HN2}) under OBC. Under such a similarity transformation, also called an imaginary gauge transformation due to its formality as a gauge transformation $c_i^\dagger\rightarrow e^{i\varphi_i}c_i^\dagger, c_i\rightarrow e^{-i\varphi_i}c_i$ with purely imaginary phases $\varphi_i$, the Hatano-Nelson model transforms into:
\begin{equation}
H_1':=U_r^{-1}H_1U_r=\sum\limits_i(t_lrc_{i+1}^\dagger c_i+t_r r^{-1}c_{i-1}^\dagger c_i).
\end{equation}

Now we generalize the Hatano-Nelson model to the complex one where the hopping coefficients are now allowed to be complex. The energy spectrum of a $N$-site complex Hatano-Nelson model under OBC is equivalent to the eigenvalues of a $N\times N$ tri-diagonal matrix:
\begin{equation}
    \begin{pmatrix}
        0 & b &  & & & \\
        a & 0 & b & & \\
         & a & \sddots & \sddots & \\
         & & \sddots & \sddots & b\\
         & & & a & 0
    \end{pmatrix},
    \label{Toeplitz}
\end{equation}
where $a, b \in \mathbb{C}$. As we show in Appendix \ref{complex_HN}, the eigenvalues of the matrix in Eq. (\ref{Toeplitz}) are \cite{smith1985}:
\begin{equation}
    \lambda_n=2(ab)^{\frac{1}{2}}\cos\frac{n\pi}{N+1}, \quad n=1,2,\cdots,N.
\end{equation}
Thus, all eigenvalues are real if and only if $(\arg (a) + \arg (b))\mod 2\pi = 0$, namely, $ab>0$. 

\begin{figure}[ht]
    \centering
    \includegraphics[width=0.98\linewidth]{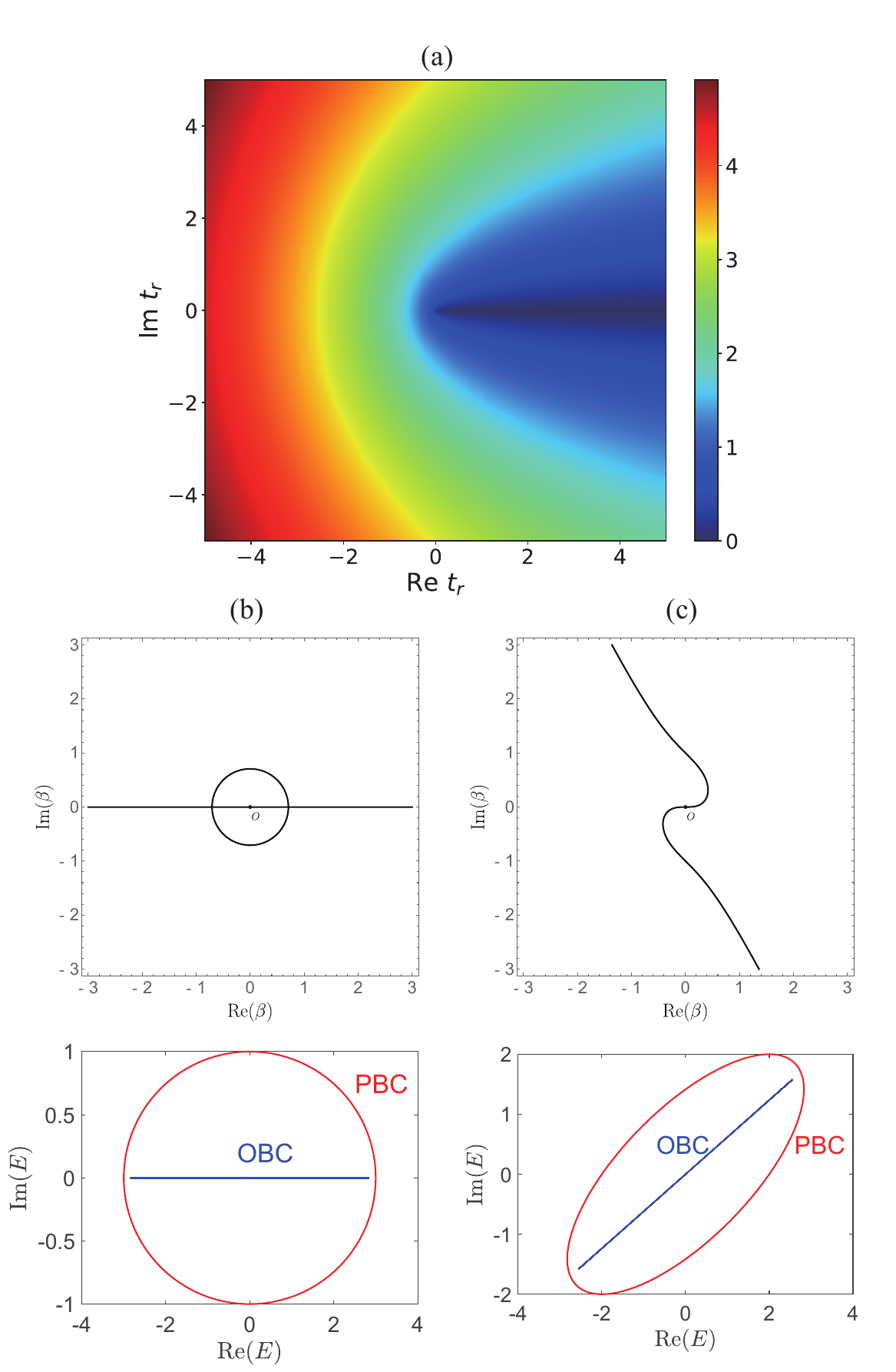}
    \caption{(a) The maximum imaginary part of the energy spectrum of the complex Hatano-Nelson model under OBC as a function of $t_r$ in the complex plane when $t_l=1$ indicates that the energy spectrum is real if and only if $t_r$ is real and positive. (b) and (c) The preimage of characteristic polynomial $a_1^{-1}(\mathbb{R})$ (above) for the Hatano-nelson model with $t_l=1$ (b) possesses the loop property for real hopping $t_r=2$ and (c) becomes a single curve for complex hopping $t_r=1+2i$; consistently, the calculated energy eigenvalues (below) exhibit a real and a complex OBC spectrum, respectively. }
    \label{max_imag}
\end{figure}

Without loss of generality, we set hopping $t_l=1$, which we can always achieve with a similarity transformation $r=t_l^{-1}$. Then, for a complex-valued hopping $t_r$, it is straightforward to see that the Hatano-Nelson model possesses a real OBC spectrum if and only if $t_r$ is purely real and positive, as shown in Fig. \ref{max_imag}(a). For more general cases, the condition for a real OBC spectrum is equivalent to $t_l t_r>0$.

Next, we compare such benchmark conditions with our theoretical criteria. First, we demonstrate the numerical results of several examples in Fig. \ref{max_imag}: when $t_r>0$, the preimage $a_1^{-1}(\mathbb{R})$ contains a loop encircling the origin, and the OBC spectrum is real, as shown in Fig. \ref{max_imag}(b); on the contrary, for a complex-valued $t_r$, the preimage will generally become a single curve and pass through the origin without forming a loop, and the OBC spectrum is no longer purely real, as in Fig. \ref{max_imag}(c).

More generally, we can also establish the condition $t_r>0$ for the loop property of the preimage $a_1^{-1}(\mathbb{R})$ analytically as follows. For real $t_r$, $a_1^{-1}(\mathbb{R})$ contains the real axis $\mathbb{R}$ and can self-intersect where the derivative of the characteristic polynomial is zero:
\begin{equation}
    a_1'(\beta)=-\frac{1}{\beta^2}+t_r=0,\quad \beta\in\mathbb{R},
\end{equation}
which has two solutions $\beta_c=\pm\sqrt{t_r}$ when $t_r>0$ and no real solution otherwise. Hence critical points exist if and only if $t_r>0$ and there must have either two or zero critical points. When the preimage contains a loop that encloses the origin, it has to cross the real axis at least twice and leaves two critical points on it, which requires $t_r>0$. On the other hand, when $t_r>0$, the asymptotic behavior of $a(\beta)$ as $\beta$ goes to infinity indicates that any curve in the preimage, with the exception for the real axis, has to start at one intersection point and end at another. Namely, together with the complex conjugate symmetry, the preimage must form a loop that encircles the origin. 

As for a complex $t_r$, the preimage cannot have the loop property. Otherwise, the loop in $a_1^{-1}(\mathbb{R})$ must cross the real axis at least twice, where we will have real points on $a_1^{-1}(\mathbb{R})$, leading to a contradiction: because $a_1(\beta)=1/\beta+t_r\beta$ cannot be real when $\beta\in\mathbb{R}$ and $t_r\in\mathbb{C}\setminus\mathbb{R}$, there cannot be any real point on $a_1^{-1}(\mathbb{R})$ except the origin. In fact, the preimage $a_1^{-1}(\mathbb{R})$ in the complex $t_r$ cases is a single curve that passes through the origin, as shown in Fig. \ref{max_imag}(c). 

From the above argument, we have established the equivalence between the condition $t_r>0$ and the loop property of the preimage $a_1^{-1}(\mathbb{R})$. As for the general case where $t_l\neq 1$, the similarity transformation $\beta\rightarrow r^{-1}\beta$ is a combination of scaling and rotation in the complex plane. Since these conformal mappings preserve angle, whether the preimage has the loop property remains intact. Thus $t_lt_r>0$ is also an equivalent condition for the loop property of $a_1^{-1}(\mathbb{R})$. Since $t_lt_r>0$ is also identical to the real OBC spectrum, we have proved the equivalence between the loop property of the preimage and the real OBC spectrum for the models in this subsection.

\subsection{The Hatano-Nelson model with next-nearest neighbor hopping}

Next, we consider the Hatano-Nelson model with an additional next-nearest-neighbor (NNN) hopping:
\begin{equation}
    H_2=\sum\limits_{i}(c_{i+1}^\dagger c_i+Ac_{i-1}^\dagger c_i+B c_{i-2}^\dagger c_i). \label{HN2}
\end{equation}
 The characteristic polynomial of this Hamiltonian is: 
\begin{equation}
a_2(\beta)=\beta^{-1}+A\beta+B\beta^2.
\end{equation}
Here, we have set the first hopping coefficient to $1$ since this can always achieve with a similarity transformation. Unless transformable to real hopping under certain similarity transformations, complex hopping will generally lead to complex OBC spectra, as we have witnessed for the pristine Hatano-Nelson model in the last subsection. Therefore, we consider real hopping $A, B \in \mathbb{R}$ here afterward. 

Like before, the real axis $\mathbb{R}$ is included in the image $a_2^{-1}(\mathbb{R})$ and may intersect with other parts of the preimage. There can only be either one or three such intersection points, where the characteristic polynomial takes a vanishing derivative:
\begin{equation}
    \frac{da_2}{d\beta}=-\frac{1}{\beta^2}+A+2B\beta=0, \quad \beta\in\mathbb{R}.
    \label{critical point}
\end{equation}

\begin{figure*}[!ht]
    \centering
    \includegraphics[width=0.7\linewidth]{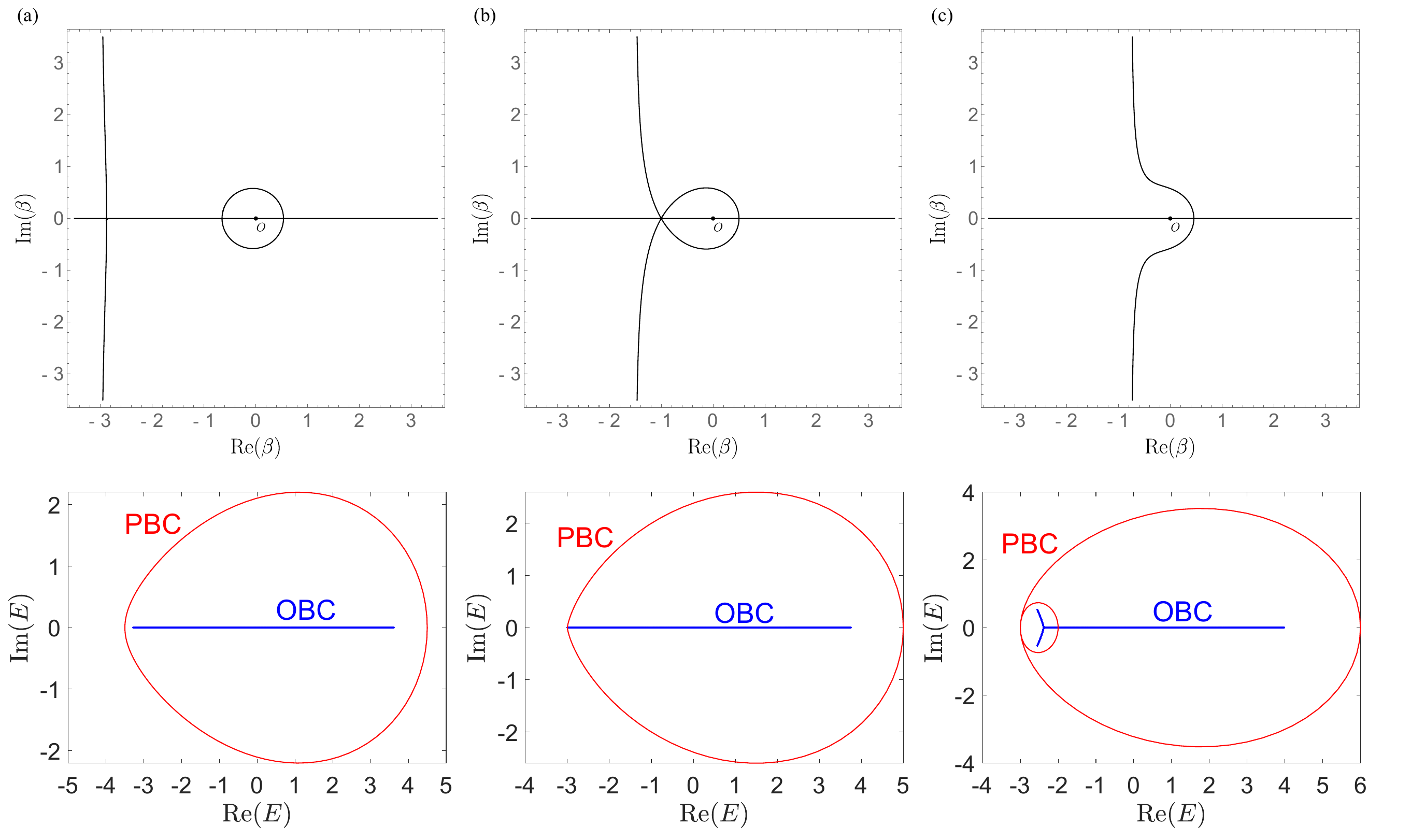}
    \caption{When the preimage of the characteristic polynomial $a^{-1}(\mathbb{R})$ (above) shows the loop property, the OBC spectrum (below) is real; otherwise, when there is an absence of a loop encircling the origin, the OBC spectrum of the corresponding non-Hermitian Hamiltonian is complex. The transition happens at $A^3=27B^2$ for $a(\beta)=1/\beta+A\beta+B\beta^2$. (a) $A=3$, $B=0.5$, $A^3> 27B^2$; (b) $A=3$, $B=1$, $A^3= 27B^2$; (c) $A=3$, $B=2$, $A^3< 27B^2$. }
    \label{net_PBC_OBC}
\end{figure*}

Indeed, we can determine the roots' conditions for a real cubic equation $a_3x^3+a_2x^2+a_1x+a_0=0$ following its discriminant $\Delta$ \cite{cohen1993}:
\begin{equation}
    \Delta=18a_3a_2a_1a_0-4a_2^3a_0+a_2^2a_1^2-4a_3a_1^3-27a_3^2a_0^2, 
\end{equation}
which suggests:
\begin{enumerate}[label=(\roman*)]
    \item if $\Delta>0$, the equation has three distinct real roots;
    \item if $\Delta<0$, the equation has a real root and two complex roots forming a complex-conjugate pair;
    \item if $\Delta=0$ and $a_2^2=3a_3a_1$, the equation has a real-root triplet; if $\Delta=0$ and $a_2^2\neq 3a_3a_1$, the equation has a single real root and a real-root doublet.
\end{enumerate}
Applying these criteria to Eq. (\ref{critical point}), we obtain:
\begin{enumerate}[label=(\roman*)]
    \item When $A^3>27B^2$, there are three intersection points;
    \item When $A^3<27B^2$, there is only one intersection point.
\end{enumerate}

According to the asymptotic behavior of $a_2(\beta)$ as $\beta\rightarrow 0 \text{ and }\infty$, only one curve in $a_2^{-1}(\mathbb{R})$ passes through $0$ and two curves extend to $\infty$, it is easy to see that when $A^3>27B^2$, there must be a loop $C$ that connects two intersection points. In addition, one can show that this loop must contain $z=0$ in its interior. The sketch of proof is: using maximum modulus principle, one can show that for an analytic function $f$ on a connected open set $U$, the imaginary part of $f$ must achieve its maximum and minimum on the boundary $\partial U$ (see Appendix \ref{Appendix: maximum modulus principle}). Let $U$ be the region contained in the loop $C$. If $z=0$ is not contained in $U$, then $a_2$ is analytic in $U$, its imaginary part should reach its maximum and minimum on $C$. However, according to the definition of $a_2^{-1}(\mathbb{R})$, $\text{Im}(a_2)$ is identically zero on $C$, thus $\text{Im}(a_2)$ vanishes identically on $U$. The Cauchy-Riemann equations \cite{lang2003}
\begin{equation}
    \frac{\partial\text{Re}(a_2)}{\partial x}=\frac{\partial\text{Im}(a_2)}{\partial y}, \quad \frac{\partial\text{Im}(a_2)}{\partial x}=-\frac{\partial\text{Re}(a_2)}{\partial y},
\end{equation}
finally lead to the constant real part of $a_2$ in $U$, which is impossible since $a_2(\beta)$ is not a constant function. Thus, besides the real axis, $a_2^{-1}(\mathbb{R})$ contains a line and a loop (only one line) when $A^3>27B^2$  
($A^3<27B^2$), see Fig. \ref{net_PBC_OBC}.

\begin{figure*}[!ht]
    \centering
    \includegraphics[width=0.98\textwidth]{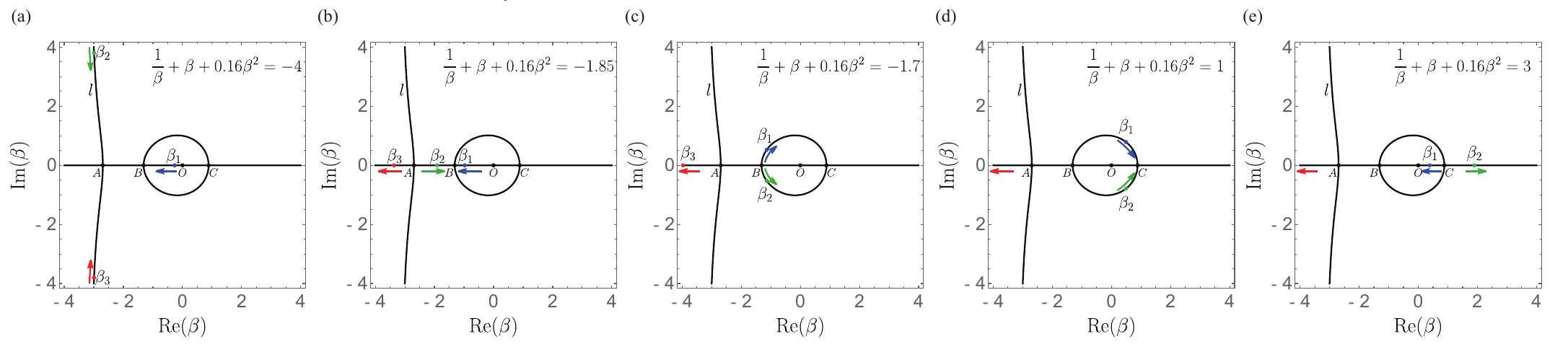}
    \caption{The evolution of the three roots $|\beta_1|\leq|\beta_2|\leq|\beta_3|$ of $1/\beta+\beta+0.16\beta^2=E$ as $E$ goes from $-\infty$ to $\infty$. The arrows denote the flow of these roots as $E$ increases: (a) $E=-4$; (b) $E=-1.85$; (c) $E=-1.7$; (d) $E=1$; (e) $E=3$. The points $A$, $B$, and $C$ correspond to the intersection points $\beta_{i3}$, $\beta_{i2}$, and $\beta_{i1}$, respectively. }
    \label{a=1_b=0.16}
\end{figure*}

Next, we show that when and only when $A^3>27B^2$, the corresponding OBC spectrum is purely real, thus indeed consistent with our criteria. Without loss of generality, we assume that $B>0$ (as the argument is similar for $B<0$), and solve Eq. (\ref{critical point}) for the three intersection points using Cardano's Formula: \begin{align}
    \beta_{i1}&=\frac{1}{u+v}, \\
    \beta_{i2}&=\Big[-\frac{u+v}{2}+\frac{i\sqrt{3}}{2}(u-v)\Big]^{-1}, \\
    \beta_{i3}&=\Big[-\frac{u+v}{2}-\frac{i\sqrt{3}}{2}(u-v)\Big]^{-1}, \\
    u&=\sqrt[3]{B+\sqrt{B^2-\frac{A^3}{27}}}, \\
    v&=\sqrt[3]{B-\sqrt{B^2-\frac{A^3}{27}}},
\end{align}
where we have taken the convention that the square root yields the branch with a positive imaginary part, and the cube root yields the branch with the largest real part. Under such a convention, it is straightforward to see that $\beta_{i3}<\beta_{i2}<0<\beta_{i1}$. Also, $a_2(\beta_{i3})<a_2(\beta_{i2})<a_2(\beta_{i1})$: since $a_2'(\beta)>0$ when $\beta_{i3}<\beta<\beta_{i2}$, we manifestly have $a_2(\beta_{i3})<a_2(\beta_{i2})$; if $a_2(\beta_{i1})\leq a_2(\beta_{i2})$, then one can plot $a_2(\beta)$ and find that some horizontal line $y=M$ can cross $a_2(\beta)$ more than three times, which leads to a contradiction and requires $ a_2(\beta_{i2})<a_2(\beta_{i1})$.

Based on these preparations, we can trace the roots of $a_2(\beta)=E$ as $E$ varies from $-\infty$ to $\infty$, as summarized in Fig. \ref{a=1_b=0.16}. We define $|\beta_1|\leq|\beta_2|\leq|\beta_3|$ following the definition in Eq. (\ref{sort}). When $E<a_2(\beta_{i3})$, $\beta_1$ is real and $\beta_2$, $\beta_3$ are two complex roots on the line $l$; then at $E=a_2(\beta_{i3})$, $\beta_2$, $\beta_3$ merge at the intersection point $\beta_{i3}$; when $a_2(\beta_{i3})<E<a_2(\beta_{i2})$, all roots are real; when $E=a_2(\beta_{i2})$, we have $\beta_1=\beta_2=\beta_{i2}$; at last, when $a_2(\beta_{i2})<E<a_2(\beta_{i1})$, $\beta_3$ is still real, but $\beta_1$ and $\beta_2$ become two complex roots on the loop.

Since we have real hopping in Eq. (\ref{HN2}), $a_2^{-1}(\mathbb{R})$ is guaranteed to possess a reflection symmetry with respect to the real axis. By tracing the roots, we have found that as long as $a_2(\beta_{i2})<E<a_2(\beta_{i1})$, $\beta_1/\beta_2$ lies in the upper/lower half part of the loop respectively, satisfying $|\beta_1|=|\beta_2|$ - the GBZ condition in Eq. (\ref{GBZ}). Thus, the loop in the preimage is a part of the GBZ.

If there are other parts of the GBZ, they must connect to this loop, like Fig. 2(b) in Ref. \cite{Zhesen2020}, as the GBZ is connected. There will be intersection points on the GBZ leading to discontinuities in the derivative of this loop. However, we can rule out such possibilities, as the loop has a continuous derivative when $A^3>27B^2$. Using $\text{Im}[a_2(\beta=x+iy)]=0$, we derive the loop's algebraic expression: 
\begin{equation}
    (x^2+y^2)(A+2Bx)=1, \quad \beta_{i2}<x<\beta_{i1},
\end{equation}
whose derivative is:
\begin{equation}
    \frac{dy}{dx}=-\frac{2Ax+2B(3x^2+y^2)}{y(2A+4Bx)},
\end{equation}
which is apparently continuous and becomes infinite only at the intersection points.

Hence, we conclude that when $A^3>27B^2$, the loop in $a_2^{-1}(\mathbb{R})$ is just the GBZ itself, and thus the OBC spectrum of the corresponding non-Hermitian Hamiltonian must be real. In the Hatano-Nelson model with an additional NNN hopping, the loop property for the preimage of the characteristic polynomial is also identical to the real OBC spectrum. 

\subsection{General 1D single-band models}

In previous subsections, we have shown the consistency of our criteria for the Hatano-Nelson model with either complex hopping or additional NNN term analytically. Here, we argue that our criteria apply to general 1D single-band models with longer-range hopping: in the following, we show that if the preimage $a^{-1}(\mathbb{R})$ of the characteristic polynomial $a(\beta)$ has the loop property, the OBC spectrum of the corresponding non-Hermitian Hamiltonians is real. 

We consider a point on the upper part of the loop in the preimage and assume that it is the $i$-th root $\beta_i(E)$ for some real energy $E$. Then, due to the reflection symmetry with respect to the real axis, there must be another root $\beta_{j}(E)$ that satisfies $\beta_j(E)=\beta_i^*(E)$, thus $|\beta_j(E)|=|\beta_i(E)|$. According to the ordering of the roots in Eq. (\ref{sort}), we can assign $j=i+1$. Since $|\beta_i(E)|\leq|\beta_{i+1}(E)|$, when $E$ varies from $-\infty$ to $\infty$, the flow of $\beta_i(E)$ follows the trajectory: close to $0$ $\rightarrow$ intersection point $A$ $\rightarrow$ upper part of the loop $\rightarrow$ intersection point $B$ $\rightarrow$ close to $0$, while $\beta_{i+1}(E)$ follows the trajectory: close to $\infty$ $\rightarrow$ intersection point $A$ $\rightarrow$ lower part of the loop $\rightarrow$ intersection point $B$ $\rightarrow$ close to $\infty$; see Fig. \ref{general_case} for illustration. Consequently, $|\beta_i(E)|\rightarrow0$ and $|\beta_{i+1}(E)|\rightarrow \infty$ when $|E|\rightarrow\infty$. According to Eq. (\ref{root_limit}), we have $i=p$, which means the loop is the trajectory of $\beta_p(E)$ and $\beta_{p+1}(E)$, and satisfies the GBZ equation $|\beta_p(E)|=|\beta_{p+1}(E)|$. In addition, the loop has a continuous derivative since it is an ``inverse" of a polynomial function. If any other parts of GBZ were connected to this loop, there would be intersection points leading to discontinuous derivatives. Following the arguments in the previous subsections, we can conclude that the loop in the preimage is just the GBZ, which in turn guarantees a real spectrum under OBC for such a non-Hermitian Hamiltonian in the presence of the loop properties. Fully consistent with the criteria, we also summarize numerical results for various non-Hermitian models with longer-range hopping in Appendix \ref{example}.

\begin{figure}[!ht]
    \centering
    \includegraphics[width=0.9\linewidth]{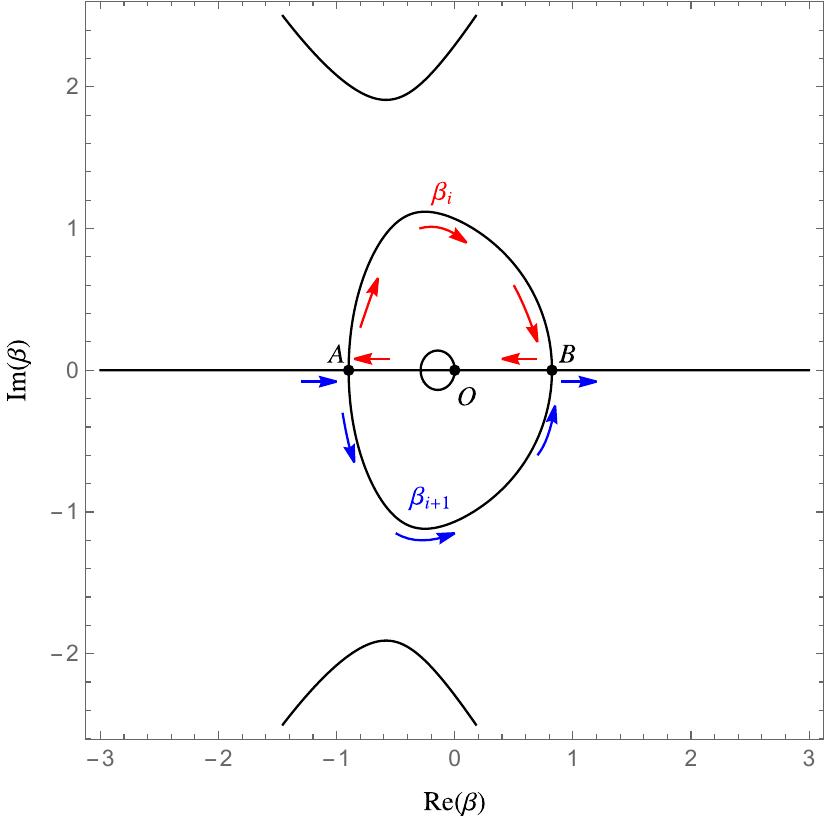}
    \caption{The trajectories of the roots move through the loop of the preimage as $E$ evolves from $-\infty$ to $\infty$. As $|\beta_i(E)|\rightarrow0$ and $|\beta_{i+1}(E)|\rightarrow\infty$ when $|E|\rightarrow\infty$, we conclude that $i=p$ and the loop is consistent with the GBZ definition.}
    \label{general_case}
\end{figure}

Note that our criteria provide a new perspective on real spectra of non-Hermitian Hamiltonians without symmetry constraints. A $\mathcal{PT}$-symmetric non-Hermitian Hamiltonian also has a real spectrum \cite{Bender2007, Bender1998}, given that such $\mathcal{PT}$ symmetry is not spontaneously broken by its eigenstates. In comparison, the non-Hermitian models presented in this work are not $\mathcal{PT}$-symmetric in general.

Let's take the Hatano-Nelson model as an example. The $\mathcal{PT}$ operator acts as \cite{PhysRevB.106.L121102}:
\begin{equation}
    \mathcal{PT}c_i(\mathcal{PT})^{-1}=c_{N+1-i}, 
\end{equation}
\vspace{-2em}
\begin{equation}
    \mathcal{PT}c_i^\dagger(\mathcal{PT})^{-1}=c_{N+1-i}^\dagger,
\end{equation}
\vspace{-2em}
\begin{equation}
    \mathcal{PT}i(\mathcal{PT})^{-1}=-i,
\end{equation}
where $N$ is the total number of lattice sites. Consequently, the action of $\mathcal{PT}$ on the Hatano-Nelson model in Eq. (\ref{Hatano_Nelson}) results in:
\begin{equation}
    (\mathcal{PT})H_1(\mathcal{PT})^{-1}=\sum\limits_i(t_r^*c_{i+1}^\dagger c_i+t_l^*c_{i-1}^\dagger c_i)=H_1^\dagger, \label{pseudo-Hermitian}
\end{equation}
suggesting that $H_1$ is not $\mathcal{PT}$-symmetric unless it is Hermitian. Despite the lack of such symmetry constraints, our study indicates that $H_1$ possesses a real spectrum as long as the condition $t_l t_r>0$ is met. Note that although $H_1$ is not $\mathcal{PT}$-symmetric, it is explicitly $\mathcal{PT}$-pseudo-Hermitian, as shown in Eq. (\ref{pseudo-Hermitian}). (Pseudo-Hermitian Hamiltonian satisfies $SHS^{-1}=H^\dagger$ \cite{Mostafazadeh2002, Mostafazadeh2010}, which does not count as a systematic symmetry conventionally.) Such pseudo-Hermiticity is expected, as a non-Hermitian Hamiltonian with a real spectrum must be pseudo-Hermitian, stated in Ref. \cite{Mostafazadeh2002}. Namely, pseudo-Hermiticity is a necessary condition for the real spectrum.

\section{Conclusion and discussion}

In summary, we have provided a sufficient and necessary condition for real spectra of 1D single-band non-Hermitian systems under OBC: the preimage of the non-Hermitian Hamiltonian's characteristic polynomial on $\mathbb{R}$ contains a closed loop with the origin inside. The condition's necessity is straightforward to establish, as 1D GBZ must be a closed loop with the origin inside. Then, we put forward analytical proofs and numerical results to establish the condition's sufficiency. For the Hatano-Nelson model with complex hopping, we obtain the energy spectrum analytically and show that both the condition for real OBC spectrum and the condition for the loop property of the preimage of the characteristic polynomial on $\mathbb{R}$ are equivalent: the multiple of the two hopping coefficients should be real and positive. Similarly, we derive analytical conditions for the real OBC spectrum and the loop property of the preimage for the Hatano-Nelson model with additional NNN hopping and establish their correspondence. We also provide numerical results on model examples with higher-order characteristic polynomials in Appendix \ref{example}. Finally, we give an argument for the condition's sufficiency in general 1D single-band models, thus providing a new perspective of real spectra in non-Hermitian systems without symmetry constraints. 

Our results help search and design 1D non-Hermitian systems with real spectra and indicate that even without symmetry constraints, a wide range of non-Hermitian systems can achieve real spectra under OBCs. Our criteria also avoid direct calculations of the energy spectra and eigenstates, the typical bottleneck and source of accumulated error in non-Hermitian systems, especially for large system sizes.

\emph{Acknowledgments} \textemdash We acknowledge helpful discussions with Yongxu Fu. We also acknowledge support from the National Key R\&D Program of China (No.2022YFA1403700) and the National Natural Science Foundation of China (No.12174008 \& No.92270102).

\appendix

\section{The OBC spectrum of complex Hatano-Nelson model}
\label{complex_HN}

As shown in the main text, we need to find the eigenvalues of the matrix in Eq. (\ref{Toeplitz}) for $a, b \in \mathbb{C}$ to obtain the OBC spectrum of the complex Hatano-Nelson model \cite{smith1985}. Let $\lambda$ be an eigenvalue of Eq. (\ref{Toeplitz}) and $v=(v_1, v_2, \cdots, v_N)^T$ be the corresponding eigenvector, then the eigenvalue equation leads to:
\begin{align*}
    -\lambda v_1+bv_2&=0, \\
    av_1-\lambda v_2+bv_3&=0, \\
    av_2-\lambda v_3+bv_4&=0, \\
    \vdots \\
    av_{N-2}-\lambda v_{N-1}+bv_N&=0, \\
    av_{N-1}-\lambda v_N&=0,
\end{align*}

By defining $v_0=v_{N+1}=0$ as the boundary condition, we can establish for the above equations the following solution:
\begin{equation}
    v_i=C_1x_1^i+C_2x_2^i \label{vi},
\end{equation}
where $C_1$ and $C_2$ are complex constants, $x_1$ and $x_2$ are the two roots of the equation:
\begin{equation}
    a-\lambda x+bx^2=0, \label{roots}
\end{equation}
which also gives $x_1 x_2=\dfrac{a}{b}$ and $x_1+x_2=\lambda/b$.

In turn, the boundary condition $v_0=v_{N+1}=0$ imposes:
\begin{align*}
    0&=C_1+C_2, \\
    0&=C_1x_1^{N+1}+C_2x_2^{N+1},
\end{align*}
which requires:
\begin{equation}
    \Big(\frac{x_1}{x_2}\Big)^{N+1}=1,
\end{equation}
\begin{equation}
    \frac{x_1}{x_2}=e^{i2\pi n/(N+1)}, \quad n=1,2,\cdots,N. \label{x1/x2}
\end{equation}

Putting together Eqs. (\ref{roots}) and (\ref{x1/x2}), we obtain:
\begin{align}
    x_1&=\Big(\frac{a}{b}\Big)^{\frac{1}{2}}e^{i\pi n/(N+1)}, \\
    x_2&=\Big(\frac{a}{b}\Big)^{\frac{1}{2}}e^{-i\pi n/(N+1)},
\end{align}
which yields the eigenvalues $\lambda=b(x_1+x_2)$ of the matrix in Eq. (\ref{Toeplitz}):
\begin{equation}
    \lambda_n=2(ab)^{\frac{1}{2}}\cos\frac{n\pi}{N+1}, \quad n=1,2,\cdots,N.
\end{equation}

Note that we have implicitly assumed that two roots of Eq. (\ref{roots}) are not equal in establishing the solution in Eq. (\ref{vi}). Otherwise Eq. (\ref{vi}) should be
\begin{equation}
    v_i=(C_1+C_2i)x^i,
\end{equation}
a solution will not hold in this case, as $v_i= (C_1+C_2i) x^i$ only yields a trivial solution $C_1=C_2=0$ under the boundary condition $v_0=v_{N+1}=0$.

\section{Algorithm for computing OBC spectrum of a general single-band non-Hermitian systems}
\label{compute_OBC}

To compute the OBC spectrum and eigenstates, one can directly diagonalize the non-Hermitian matrix. However, such eigenvalue problems' results are highly susceptible to accumulated numerical errors as the non-Hermitian model's lattice size $N$ becomes large and may be inaccurate. Here, we employ an algorithm valid for a general single-band non-Hermitian Hamiltonian based on the GBZ theory \cite{Beam1993}. We seek all values of $E=a(\beta)=\sum_{n=-p}^qa_n\beta^n$ that satisfy $|\beta_p(E)|=|\beta_{p+1}(E)|$. Assume that
\begin{equation}
    \beta_p=\beta_0e^{i\phi_l}, \quad \beta_{p+1}=\beta_0e^{-i\phi_l},
\end{equation}
where $\phi_l=l\pi/(M+1)$, $l=1,2,\cdots,M$, and $M$ is the number of points on the energy spectrum. We set $M=10000$, sufficient for all intentions and purposes throughout this work. 

Since we have: 
\begin{subequations}
\begin{eqnarray}
    E =a(\beta_p)= \sum\limits_{n=-p}^q a_n\beta_0^ne^{in\phi_l}, \label{Ea}\\
    E =a(\beta_{p+1})= \sum\limits_{n=-p}^qa_n\beta_0^ne^{-in\phi_l}, \label{Eb}
\end{eqnarray}
\end{subequations}
we obtain an equation for $\beta_0$:
\begin{equation}
    \sum\limits_{k=-p}^qa_n\beta_0^n\sin\frac{nl\pi}{M+1}=0. \label{algorithm}
\end{equation}

We can compute the  OBC spectrum for each $l$ in $1,2,\cdots,M$ following the algorithmic steps:
\begin{enumerate}
    \item [(i).] Solve Eq. (\ref{algorithm}) and obtain $p+q$ roots $\beta_0$.
    \item [(ii).] For each $\beta_0$ obtained above, check whether $\beta_a=\beta_0e^{i\phi_l}$ and $\beta_b=\beta_0e^{-i\phi_l}$ are $\beta_p$ and $\beta_{p+1}$. This is accomplished as follows:
    \item [(iia).] Insert $\beta_0$ into Eq. (\ref{Ea}) or Eq. (\ref{Eb}) to obtain $E$, substitute this $E$ back into $E=\sum_{n=-p}^qa_n\beta^n$, solve this equation to determine the remaining $(p+q-2)$ roots.
    \item [(iib).] If $|\beta_a|=|\beta_b|=|\beta_p|=|\beta_{p+1}|$, then $E$ is on the OBC spectrum; otherwise, $E$ is not on the OBC spectrum. Return to (iia) until all $\beta_0$ have been checked. 
    \item [(iii).] Return to step (i) for the next $l$ until $l=M$. 
\end{enumerate}

\section{Some theorems of complex analysis} \label{Appendix: maximum modulus principle}

The maximum modulus principle states \cite{lang2003}:

\begin{theorem}
    Let $U$ be an open region of $\mathbb{C}$, $f$ is an analytic function on $U$. If $z_0\in U$ is maximum point for $|f|$, that is $|f(z_0)|\geq |f(z)|$ for all $z\in U$, then $f$ is constant on $U$. Namely, $|f|$ only achieves its maximum on the boundary $\partial U$.
\end{theorem}

\begin{proof}
    Choose $\delta>0$ so that the disc  $D(z_0,\delta)\subset U$. Choose $0<r<\delta$ and then use Cauchy integral formula: 
    \begin{equation}
        f(z_0)=\frac{1}{2\pi i}\int_{|z-z_0|=r}\frac{f(z)}{z-z_0}dz.
    \end{equation}
    
    Rewrite the above formula in terms of the parametrisation $z=z_0+re^{i\theta}, 0\leq \theta\leq 2\pi$:
    \begin{align}
        f(z_0)&=\frac{1}{2\pi i}\int_0^{2\pi}\frac{f(z_0+re^{i\theta})}{re^{i\theta}}ire^{i\theta}d\theta \\
        &=\frac{1}{2\pi}\int_0^{2\pi}f(z_0+re^{i\theta})d\theta 
    \end{align}
    The definition of $z_0$ implies that:
    \begin{align}
        |f(z_0)|&\leq \frac{1}{2\pi}\int_0^{2\pi}|f(z_0+re^{i\theta})|d\theta\\
        &\leq \frac{1}{2\pi}\int_0^{2\pi}|f(z_0)|d\theta=|f(z_0)| \label{inequality}
    \end{align}
    since $|f(z_0+re^{i\theta})|\leq |f(z_0)|, \forall \theta$. 
    
    Then Eq. (\ref{inequality}) indicates that $|f(z_0+re^{i\theta})|= |f(z_0)|$ for all $\theta$ and $0<r<\delta$. This means that $f(z)$ is locally constant near $z=z_0$. Then apply the following theorem \cite{lang2003}:
    
    \begin{theorem}[Identity theorem]
        Let $\Omega$ be an open region of $\mathbb{C}$. If $f$ and $g$ are analytic on $\Omega$ and $\{z\in\Omega:f(z)=g(z)\}$ has an limit point in $\Omega$, then $f\equiv g$ in $\Omega$.
    \end{theorem}
     
    $f(z)$ is globally constant on $U$.
\end{proof}

Apply the maximum modulus principle to $e^{-if(z)}$ and $e^{if(z)}$ respectively, we obtain:

\begin{theorem}
    Let $U$ be an open region of $\mathbb{C}$, $f$ is an analytic function on $U$. Then $\text{Im}(f)$ achieves both its maximum and minimum on the boundary $\partial U$.
\end{theorem}

\section{Numerical results}
\label{example}

In the main text, we have analyzed the complex Hatano-Nelson model and Hatano-Nelson model with next-nearest-neighbor hopping in an analytical manner, whose $p=1$, $q=1$ and $p=1$, $q=2$, respectively. We have also numerically examined more complicated one-band non-Hermitian models with higher values of $p$ and $q$ and found that the equivalence between the loop property of the preimage and real OBC spectrum holds in general cases. 

For example, we demonstrate the numerical results of two models:
\begin{eqnarray}
H_3 &=& \sum_n (\alpha_3 c^\dagger_{n+2} + 3 c^\dagger_{n+1} + 3.2 c^\dagger_{n-1}+ c^\dagger_{n-2}+ 0.5 c^\dagger_{n-3}) c_n \nonumber \\
H_4 &=& \sum_n (c^\dagger_{n+1} + 1.5 c^\dagger_{n-1} + 0.3 c^\dagger_{n-2}+ 0.2 c^\dagger_{n-3}+ \alpha_4 c^\dagger_{n-4}) c_n, \nonumber \\
\label{app:h3h4}
\end{eqnarray}
whose characteristic polynomials are:
\begin{align}
a_3(\beta)&=\alpha_3\beta^{-2}+3\beta^{-1}+3.2\beta+\beta^2+0.5\beta^3, \\
a_4(\beta)&=\beta^{-1}+1.5\beta+0.3\beta^2+0.2\beta^3+\alpha_4\beta^4,
\end{align}
and $(p, q)=(2, 3)$ and $(1, 4)$, respectively.  

We summarize typical results in Figs. \ref{numerical_example1}-\ref{numerical_example2}: whether the preimage $a_i^{-1}(\beta)$ has the loop property or not is entirely consistent with the corresponding OBC spectrum being purely real or retaining complex branches, as we vary the hopping coefficients $\alpha_3$ and $\alpha_4$. These results are consistent with our conclusions in the main text.

\begin{figure*}[!ht]
    \centering
    \includegraphics[width=0.98\textwidth]{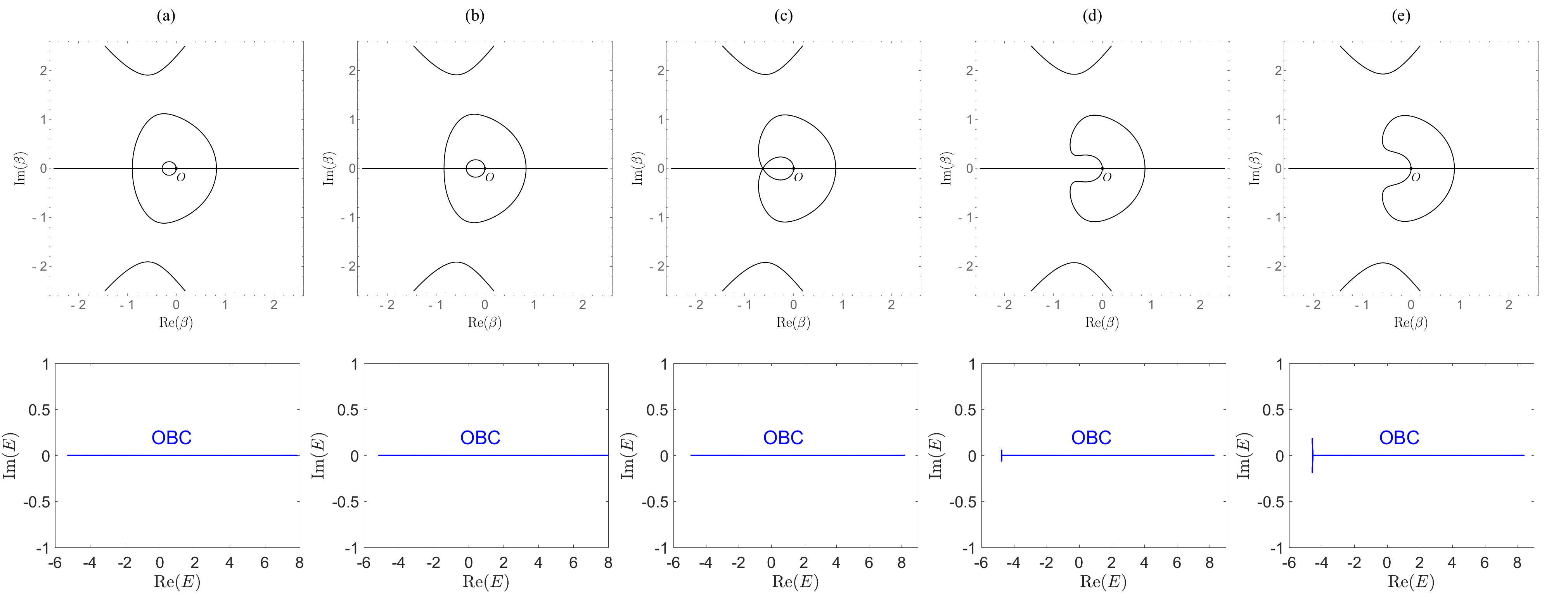}
    \caption{The preimage $a^{-1}(\mathbb{R})$ (above) and the OBC spectrum (below) for the single-band non-Hermitian tight-binding Hamiltonian $H_3$ in Eq. (\ref{app:h3h4}) with the characteristic polynomial  $a_3(\beta)=\alpha_3\beta^{-2}+3\beta^{-1}+3.2\beta+\beta^2+0.5\beta^3$ and various values of the hopping coefficient $\alpha_3$: (a) $\alpha_3=0.4$, (b) $\alpha_3=0.5$, (c) $\alpha_3=0.628$, (d) $\alpha_3=0.7$, and (e) $\alpha_3=0.8$.}
    \label{numerical_example1}
\end{figure*}

\begin{figure*}[!ht]
    \centering
    \includegraphics[width=0.98\textwidth]{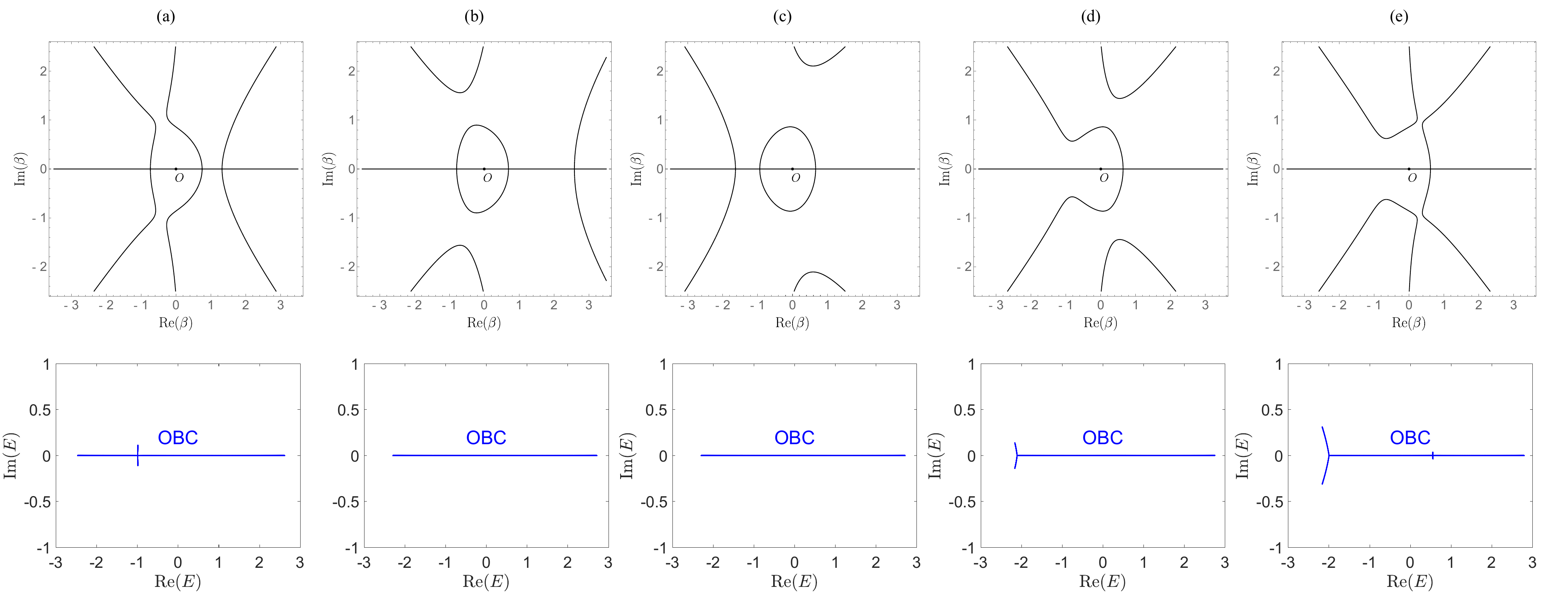}
    \caption{The preimage $a^{-1}(\mathbb{R})$ (above) and the OBC spectrum (below) for the single-band non-Hermitian tight-binding Hamiltonian $H_4$ in Eq. (\ref{app:h3h4}) with the characteristic polynomial  $a_4(\beta)=\beta^{-1}+1.5\beta+0.3\beta^2+0.2\beta^3+\alpha_4\beta^4$ and various values of the hopping coefficient $\alpha_4$: (a) $\alpha_4=-0.3$, (b) $\alpha_4=-0.1$, (c) $\alpha_4=0.1$, (d) $\alpha_4=0.3$, (e) $\alpha_4=0.6$.}
    \label{numerical_example2}
\end{figure*}

\clearpage

\bibliography{refs}

\end{document}